\documentclass{acm_proc_article-sp}
\pdfoutput=1
\usepackage{color}
\usepackage{hyperref}

\usepackage[vlined,ruled, linesnumbered]{algorithm2e}

\usepackage{amsmath}  
\usepackage[utf8]{inputenc}
\usepackage{amssymb}
\usepackage{xspace}

\usepackage{float}
\usepackage{colonequals}

\renewcommand{\paragraph}[1]{\noindent{\bf #1.}}

\newenvironment{code}{\sf\vspace*{-12pt}}{\vspace*{-2pt}}

\def\qed {{                
   \parfillskip=0pt        
   \widowpenalty=10000     
   \displaywidowpenalty=10000  
   \finalhyphendemerits=0  
   \leavevmode             
   \unskip                 
   \nobreak                
   \hfil                   
   \penalty50              
   \hskip.2em              
   \null                   
   \hfill                  
   $\square$
   \par}}                  

\newcommand{\Prop}{\ensuremath{\mathsf{Prop}\xspace}}
\newcommand{\IRI}{\ensuremath{\mathsf{IRI}\xspace}}
\newcommand{\Inv}[1]{\ensuremath{\text{\sf{}\^{}}#1}}
\newcommand{\InvProp}{\Inv{\Prop}}
\newcommand{\prop}{\ensuremath{\mathsf{prop}\xspace}}
\newcommand{\node}{\ensuremath{\mathsf{node}\xspace}}
\newcommand{\Nodes}{\ensuremath{\mathsf{Nodes}\xspace}}
\newcommand{\Edges}{\ensuremath{\mathsf{Edges}\xspace}}
\newcommand{\val}{\ensuremath{\mathsf{val}\xspace}}
\newcommand{\blank}{\ensuremath{\mathsf{\_b}\xspace}}
\newcommand{\neigh}{\ensuremath{\mathsf{neigh}\xspace}}

\newcommand{\AND}{\ensuremath{\mathbin{\mathsf{AND}}\xspace}}
\newcommand{\ShapeExpr}{\ensuremath{\mathsf{ShapeExpr}}\xspace}
\newcommand{\EmptyShape}{\ensuremath{\mathsf{EmptyShape}}\xspace}
\newcommand{\TripleConstraint}{\ensuremath{\mathsf{TripleConstraint}}\xspace}
\newcommand{\Property}{\ensuremath{\mathsf{Property}}\xspace}
\newcommand{\InvProperty}{\ensuremath{\mathsf{InvProperty}}\xspace}
\newcommand{\ValueClass}{\ensuremath{\mathsf{ValueClass}}\xspace}
\newcommand{\SomeOfShape}{\ensuremath{\mathsf{SomeOfShape}}\xspace}
\newcommand{\GroupShape}{\ensuremath{\mathsf{GroupShape}}\xspace} 
\newcommand{\RepetitionShape}{\ensuremath{\mathsf{RepetitionShape}}\xspace} 
\newcommand{\Cardinality}{\ensuremath{\mathsf{Cardinality}}\xspace}
\newcommand{\MinCardinality}{\ensuremath{\mathsf{MinCardinality}}\xspace}
\newcommand{\MaxCardinality}{\ensuremath{\mathsf{MaxCardinality}}\xspace}
\newcommand{\unbound}{\ensuremath{\mathsf{unbound}}\xspace} 
\newcommand{\ShapeExpressionSchema}{\ensuremath{\mathsf{ShapeExprSchema}}\xspace}
\newcommand{\ValueSet}{\ensuremath{\mathsf{ValueSet}}\xspace} 
\newcommand{\ShapeConstr}{\ensuremath{\mathsf{ShapeConstr}}\xspace}
\newcommand{\ShapeLabel}{\ensuremath{\mathsf{ShapeLabel}}\xspace}
\newcommand{\AtomicConstr}{\ensuremath{\mathsf{AtomicConstr}}\xspace}

\newcommand{\ExtraPropSet}{\ensuremath{\mathsf{ExtraPropSet}}\xspace}
\newcommand{\ShapeDefinition}{\ensuremath{\mathsf{ShapeDefinition}}\xspace}
\newcommand{\CLOSED}{\ensuremath{\mathsf{CLOSED}}\xspace}

\newcommand{\definition}{\ensuremath{\mathsf{definition}}\xspace}
\newcommand{\expr}{\ensuremath{\mathsf{expr}}\xspace}
\newcommand{\Shapes}{\ensuremath{\mathsf{Shapes}}\xspace}
\newcommand{\negatedShapes}{\ensuremath{\mathsf{negated\text{\sf{}-}shapes}}\xspace}

\newcommand{\TCons}{\ensuremath{\mathsf{TCons}}\xspace}
\newcommand{\extra}{\ensuremath{\mathsf{extra}}\xspace}
\newcommand{\open}{\ensuremath{\mathsf{open}}\xspace}
\newcommand{\ShDef}{\ensuremath{\mathsf{ShDef}}\xspace}
\newcommand{\Consumers}{\ensuremath{\mathsf{Consumers}}\xspace}

\newcommand{\witness}{\ensuremath{\mathsf{witness}}\xspace}
\newcommand{\Neigh}{\ensuremath{\mathsf{Neigh}}\xspace}
\newcommand{\edge}{\ensuremath{\mathsf{edge}}\xspace}
\newcommand{\Expr}{\ensuremath{\mathsf{Expr}}\xspace}

\newcommand{\NegatedShapes}{\ensuremath{\mathsf{NegatedShapes}}\xspace}

\newcommand{\Neg}[1]{\ensuremath{\mathsf{!}#1}}

\newcommand{\typing}{\ensuremath{\mathsf{typing}}\xspace}

\newcommand{\TUC}{\mathsf{TUC}}
\newcommand{\lw}{\mathsf{lw}}
\newcommand{\requires}{\mathsf{requires}}
\newcommand{\ULW}{\mathsf{UnchLW}}

\newcommand{\lwcert}{\mathsf{lw}^{\mathsf{cert}}}
\newcommand{\typingcert}{\mathsf{typing}^{\mathsf{cert}}}
\newcommand{\toCheck}{\mathsf{toCheck}}
\newcommand{\typinghyp}{\typing^{\mathsf{hyp}}}
\newcommand{\lwhyp}{\lw^{\mathsf{hyp}}}
\newcommand{\ml}{\mathopen{\{\hspace{-0.2em}|}}
\newcommand{\mr}{\mathclose{|\hspace{-0.2em}\}}}

\SetNlSty{textrm}{}{:}
\DontPrintSemicolon
\SetKw{KwAnd}{and}
\SetKw{KwOr}{or}
\SetKw{KwNot}{not}
\SetKw{KwCont}{continue}
\SetKw{KwLet}{let}

\usepackage[noconfig]{ntheorem}
\newtheorem{theorem}{Theorem}[section]

\newtheorem{lemma}[theorem]{Lemma}
\newtheorem{corollary}[theorem]{Corollary}

\theorembodyfont{\normalfont}
\newtheorem{example}[theorem]{Example}
\newtheorem{DEF}{Definition}

\theoremheaderfont{\sc}
\theorembodyfont{\itshape}

\begin{document}

\title{Shape Expressions Schemas}

\numberofauthors{4} 
\author{
\alignauthor
Iovka Boneva\\
\affaddr{Univ. Lille - CRIStAL - F-59000 Lille, France}
\email{iovka.boneva@univ-lille1.fr}
\alignauthor
Jose E. Labra Gayo\\
\affaddr{University of Oviedo, Spain}\\
\email{labra@uniovi.es}\\
\and
\alignauthor
Eric G. Prud'hommeaux\\
\affaddr{W3C}\\
\affaddr{Stata Center, MIT}
\email{eric@w3.org}
\alignauthor
S\l{}awek Staworko\\
\affaddr{Univ. Lille - CRIStAL - F-59000 Lille, France}\\
\affaddr{University of Edinburgh, UK}\\
{\email{slawomir.staworko@inria.fr}}
}
\maketitle

\begin{abstract}

We present Shape Expressions (ShEx), an expressive schema language for RDF designed to provide a high-level, user friendly syntax with intuitive semantics.
ShEx allows to describe the vocabulary and the structure of an RDF graph, and to constrain the allowed values for the properties of a node.
It includes an algebraic grouping operator, a choice operator, cardinalitiy constraints for the number of allowed occurrences of a property, and negation.
We define the semantics of the language and illustrate it with examples.
We then present a validation algorithm that, given a node in an RDF graph and a constraint defined by the ShEx schema, allows to check whether the node satisfies that constraint.
The algorithm outputs a proof that contains  trivially verifiable associations of nodes and the constraints that they satisfy.
The structure can be used for complex post-processing tasks, such as transforming the RDF graph to other graph or tree structures, verifying more complex constraints, or debugging (w.r.t. the schema).
We also show the inherent difficulty of error identification of ShEx.

\end{abstract}

\section{Introduction}
\label{sec:intro}

RDF's distributed graph model encouraged adoption for publication and manipulation of e.g. social and biological data.
Coding errors in data stores like DBpedia have largely been handled in a piecemeal fashion with no formal mechanism for detecting or describing schema violations.
Extending uptake into environments like medicine, business and banking requires structural validation analogous to what is available in relational or XML schemas.

While OWL ontologies can be used for limited validation, they are generally used for formal models of reusable classes and predicates describing objects in some domain.
Applications typically consume and produce graphs composed of precise compositions of such ontologies.
A company's human resources records may leverage terms from FOAF and Dublin Core, but only certain terms, and composed in specific structures.
We would no more want to impose the constraints of a single human resources application suite on FOAF and Dublin Core than we would want to assert that such applications need to consume all ontologically valid permutations of FOAF and Dublin Core entities.
Further, open-world constraints on OWL ontologies make it impossible to use conventional OWL tools to e.g. detect missing properties.
Shape expressions define structural constraints (arc labels, cardinalities, datatypes, etc.) to provide a schema language in which it is easy to mix terms from arbitrary ontologies.

A schema language for any data format has several uses:
communicating to humans and machines the "shape" of input/output data;
enabling machine-verification of data on production, publication, or consumption;
driving query and input interfaces;  static analysis of queries.
In this, ShEx provides a similar role as relational and XML schemas.
Declarative switches in ShEx schemas (c.f. $\CLOSED$ and {\sf EXTRA} defined below) make is useful for tightly controlled environments intolerant of extra assertions as well as linked open data media which promises a core data structure peppered with arbitrary extra triples.

ShEx validates nodes in a graph against a schema construct called a \emph{shape}.
In XML, validating an element against an XML Schema \cite{xml-schema} type or element or Relax NG \cite{relaxng} production recursively tests nested elements against constituent rules.
In ShEx, validating a \emph{focus node} in a graph against a shape recursively tests the nodes which are the subjet or object of triples constrained in that shape.

ShEx was designed to provide a sound and coherent language without variables.
The compact syntax and the JSON syntax are intended to enable trivial authoring and parsing by people and machines.
The core language's balance between expressivity and complexity is supplemented by an extension mechanism which enables more expressive semantic actions, acting like Schematron~\cite{schematron} rules embedded in XML Schema or Relax NG.

\paragraph{Previous Work and Contributions}
In \cite{semantics2014} we presented a first version of shape expressions, in which we used a conjunction operator instead of grouping, and recursion was not discussed.
In \cite{icdt2015} we studied the complexity of validation of ShEx in absence of negation, and for closed shape definitions only. 
We showed that in general the complexity is NP-complete, identified tractable fragments, and proposed a validation algorithm for a restriction called deterministic single-occurrence shape expressions.
In the current paper, we present an enhanced version of ShEx schemas, including the $\CLOSED$ and {\sf EXTRA} modifiers, negation and well-defined recursion, value sets, and conjunction in value class constraints. We present the semantics of ShEx (Section~\ref{sec:formal}) independently on regular bag expressions \cite{icdt2015}, which we believe makes it easier to understand by a larger community. We also present a full validation algorithm, as well as guidelines for its efficient implementation (Section~\ref{sec:validation}). We finally show that error identification for ShEx is a hard problem, even if only the basic constructs of the language are used (Section~\ref{sec:errors}). ShEx has several open-source implementations\footnote{Implementations: \url{http://shex.io/\#platforms/};.} not discussed here because of space constraints, two of which can be tested online.%
\footnote{Demo: \url{http://rdfshape.herokuapp.com/}.}%
\footnote{Demo: \url{http://www.w3.org/2013/ShEx/FancyShExDemo}.}

\section{Primer}

The following examples illustrate several features of shape
expressions. Let {\sf is:} be a namespace prefix for a widely-used issue tracking ontology, and {\sf foaf:} and {\sf
  xsd:} are the standard FOAF and XSD prefixes, respectively. {\sf ex:}
is the namespace prefix for some example instance data which we test against our example schema. ShEx
uses the same conventions as Turtle~\cite{TurtleRDF} with relative and absolute IRIs
enclosed in {\sf <} {\sf >} and prefixed names. 
as a shorthand notation for IRIs. As a convention in this primer, we
will use relative IRIs as shape identifiers for our application schema.

\paragraph{Running Example}
We give an example of a ShEx
schema for data manipulated by a bug tracker. It describes five node shapes of
interest. A shape describes constraints on the graph edges touching a
particular node, called the \emph{focus node}. The shape {\sf
  <TesterShape>} describes the constraints for a node that represents
a tester. It contains two \emph{triple constraints}, that associate a property with a required value. 
A tester node must have one {\sf foaf:name} property the
object of which is a literal string value, and one {\sf is:role}
property the object of which is an IRI. The two triple constraints are \emph{grouped}, using the comma (,) operator, indicating that the node must have triples that satisfy all parts of the grouping.

\begin{code}
  \begin {tabbing}
    \hspace{0.6cm} \= \hspace{0.4cm} \=\\
    <TesterShape> \{\\
    \> foaf:name xsd:string, is:role IRI \}
  \end {tabbing}
\end{code}

The shape {\sf <ProgrammerShape>} requires that a node has one {\sf
  foaf:name} that is a string, one {\sf is:experience} property which
value is one among {\sf is:senior} and {\sf is:junior}. 

\vspace{-0.1cm}
\begin{code}
  \begin {tabbing}
    \hspace{0.6cm} \= \hspace{0.4cm} \=\\
    <ProgrammerShape> \{\\
    \> foaf:name xsd:string,\\
    \> is:experience (is:senior is:junior) \}
  \end {tabbing}
\end{code}

A user ({\sf <UserShape>}) has either a {\sf foaf:givenName} and a
{\sf foaf:lastName}, or a single {\sf foaf:name}, and all these are
strings. The choice between the two possibilities is expressed using
the \emph{some-of} operator, written as a vertical bar {\sf
  |}. Moreover, a user has zero or more {\sf is:affectedBy} properties which value satisfies {\sf <IssueShape>}. The
\emph{repetition constraint} "zero or more" is expressed by the {\sf
  *} symbol.
A {\sf <ClientShape>} is required to have only a {\sf is:clientNumber} that is an
integer.

\vspace{-0.5cm}
\begin{code} 
 \begin {tabbing}
    \hspace{0.3cm} \= \hspace{0.3cm} \= \hspace{7,2cm} \= \\
    <UserShape> \{\\
    \> ( (foaf:givenName xsd:string, foaf:lastName xsd:string)\\
    \>\> | foaf:name xsd:string ) , \\
    \> is:affectedBy @<IssueShape> *\}\\
    <ClientShape> \{\\
    \> is:clientNumber xsd:integer \}
  \end{tabbing}
\end{code}

Finally, we describe {\sf <IssueShape>}. 
The shape {\sf <IssueShape>} specifies that the node needs to
have one {\sf is:reportedBy} property, whose objects satisfies both
{\sf <UserShape>} and {\sf <Client\-Shape>}. Also, an issue
node has to be {\sf is:reproducedBy} one tester and {\sf
  is:reproducedBy} one or more programmers. The "one or more" is written using the + sign. Additionally, an 
issue must have one or more \emph{inverse} (i.e. incoming) {\sf is:affectedBy}
arc whose subject is a user; the incoming property requirement is expressed by the hat ($\Inv$) preceding the
property. The annotations in comments (after the \#) on the right of {\sf <IssueShape>} will be used later on. The {\sf EXTRA} modifier is explained below.

\begin{code} 
 \begin {tabbing}
    \hspace{0.3cm} \= \hspace{0.3cm} \= \hspace{7,2cm} \= \\
    <IssueShape> \\
    \> EXTRA is:reproducedBy \{ \>\> \#$\mathsf{TCons}_0$\\
    \> is:reportedBy @<UserShape> AND @<ClientShape>,  \>\> \#$\mathsf{C}_1$\\
    \> is:reproducedBy @<TesterShape> , \>\> \#$\mathsf{C}_2$\\
    \> is:reproducedBy @<ProgrammerShape> +, \>\> \#$\mathsf{C}_3$\\
    \> $\Inv$is:affectedBy @<UserShape> + \} \>\> \#$\mathsf{C}_4$
  \end {tabbing}
\end{code}

Here is a portion of RDF data conforming to the ShEx schema. On
the right of every subject node we list the shapes satisfied by that
node, and that allow to witness that {\sf ex:issue1} and {\sf ex:issue2}
satisfy {\sf <IssueShape>}. The $\edge_j$ annotations will be used later on.

\begin{code}
 \begin {tabbing}
    \hspace{0.3cm} \= \hspace{4,1cm} \= \hspace{0.5cm} \= \hspace{0.2cm} \= \hspace{2cm} \= \\ 
    ex:issue1  \>\>\>\> \#<IssueShape> \\
    \> is:reportedBy ex:fatima ; \>\>\>\> \#$\edge_1$ \\
    \> is:reproducedBy ex:ren ; \>\>\>\> \#$\edge_2$ \\
    \> is:reproducedBy ex:noa ; \>\>\>\> \#$\edge_3$ \\
    \> is:reproducedBy ex:emin ; \>\>\>\> \#$\edge_4$\\
    \> is:dueDate "15/12/2015"\^{}xsd:date . \>\>\>\> \#$\edge_5$ \\
    ex:issue2 is:reportedBy ex:emin ;  \>\>\>\> \#<IssueShape> \\
    \> is:reproducedBy ex:ren, ex:noa, ex:shristi . \\
    ex:ren foaf:name "Ren Traore" ; \>\>\> \#<TesterShape> \\
    \> is:role is:integration ; \\
    \> is:affectedBy ex:issue2 . \\
    ex:noa foaf:name "Noa Salma" ; \>\>\> \#<ProgrammerShape> \\
    \> is:experience is:senior ; \\
    \> foaf:mbox "noa@mail.com" . \\
    ex:shristi foaf:name "Shristi Li" ; \>\>\> \#<ProgrammerShape>  \\
    \> is:experience is:junior . \\
    ex:fatima is:clientNumber 1 ;  \>\> \#<UserShape><ClientShape> \\
    \> foaf:givenName "Fatima"; \\
    \> foaf:lastName "Smith" . \\
    ex:emin is:clientNumber 2 ; \>\> \#<UserShape><ClientShape> \\
    \> foaf:name "Emin V. Petrov" ; \\
    \> is:affectedBy ex:issue1 .  \>\>\>\> \#$\edge_6$
  \end{tabbing}
\end{code}

By default, the properties that are mentioned in the shape definition
are allowed to appear only conforming to the corresponding
constraints. For instance, a programmer cannot have a {\sf
  is:experience} property whose value is not one among {\sf is:senior}
and {\sf is:junior}. This default behaviour can be changed using the
{\sf EXTRA} modifier, that takes as parameter a list of
properties. For all properties declared {\sf EXTRA}, triples additional to those satisfying the triple constraints are allowed in any number, provided that their object does not satisfy the
constraints being mentioned. {\sf <IssueShape>} has the {\sf EXTRA
  is:reportedBy} modifier, which allows for {\sf is:reportedBy}
properties whose value is \emph{neither a tester nor a
  programmer}. Without the {\sf EXTRA} modifier, the triple {\sf
  ex:issue1 is:reproducedBy ex:emin} would have prevented {\sf
  ex:issue1} from satisfying {\sf <IssueShape>}. On the other hand,
even in presence of {\sf EXTRA}, an issue shape can still be {\sf
  is:reproducedBy} only one tester.

Regarding the properties not mentioned in the shape definition, they
are by default in any number and with unconstrained values.  For
instance, a programmer can have a {\sf foaf:mbox}
property, and an issue can have a {\sf is:dueDate}.
This default behaviour can be tuned too: the $\CLOSED$ modifier
forbids all properties that are not mentioned in the shape definition,
and the $\Inv\CLOSED$ modifier forbids all inverse properties not
mentioned.  The modifiers $\CLOSED$ and {\sf EXTRA} can be combined,
in order to forbid all non mentioned properties, except for those that
are arguments of {\sf EXTRA}.

ShEx schemas also include a negation operator $\Neg$. For instance, one
could define an issue having low impact as an issue that is possibly
reproduced and reported several times, but not by a client.

\begin{code}
  \begin {tabbing}
    \hspace{0.6cm} \= \hspace{0.4cm} \=\\
    <LowImpactIssueShape> \{\\
    \> is:reportedBy \Neg{}@<ClientShape> *, \\
    \> is:reproducedBy \Neg{}@<ClientShape> * \}
  \end{tabbing}
\end{code}
The Running Example can be tested online\footnote{Demo of the Running Example: \url{http://goo.gl/EhlksQ}.}

\paragraph{Repeated Properties}
Much of the design of ShEx arises from meeting use cases with multiple triple constraints on the same property (or \emph{Repeated Properties}) as seen above in {\sf <IssueShape>} which requires {\sf is:re\-pro\-ducedBy} arcs for at least one tester and at least one programmer.
These are quite common in clinical informatics where generic properties relate observations in specific ways, but in fact emerge any time a schema or convention re-purposes generic predicates.
Because RDF is a graph, the same node may be involved in the validation of multiple triples where that node matches multiple shapes.
For instance, consider a small modification of the Running Example, by adding the triple
$$
\textsf{ex:shristi is:role is:integration .} 
$$
Then {\sf ex:shristi} satisfies both {\sf <ProgrammerShape>} and {\sf <Tes\-terShape>}. 
Validating {\sf ex:issue2} as a {\sf <IssueShape>} requires {\sf ex:shristi} to be seen as a programmer, as {\sf ex:ren} satisfies only {\sf <TesterShape>}.
A less nuanced interpretation would treat all repeated properties in a shape as a {\emph{conjunction of constraints}, in which case {\sf ex:issue2} would not match {\sf <IssueShape>} because {\sf ex:ren} would have to match both {\sf <ProgrammerShape>} and {\sf <TesterShape>}.

\paragraph{Extension Mechanism and JSON Syntax}
ShEx schemas have two additional features that we do not
present here because of space constraints. 
The \emph{extension mechanism} presented also in \cite{semantics2014},
allows to decorate schemas with executable code, called \emph{semantic
  actions}\footnote{Extension mechanism: \url{http://shex.io/primer/\#semact}.},
which are useful for more expressive validation.
For instance, a semantic action might be used to check value ranges or invoke a service to test membership in an external value set:
\begin{code}
  \begin {tabbing}
    \hspace{0.4cm} \= \hspace{0.4cm} \=\\
    <ClientShape> \{\\
    \> is:clientNumber xsd:integer\\
    \> \> \emph{\%ex:memberOf\{ex:OurClients\%\}}\ \}
  \end {tabbing}
\end{code}
Semantic actions can also produce output to 
e.g. construct XML fitting some XML schema.

The examples in this section are all presented with what is called the
\emph{compact syntax} or \emph{ShExC}, that is intended to be processed by
humans. ShEx schemas are easier to programmatically produce and manipulate using the \emph{JSON syntax}
\footnote{JSON syntax: \url{http://shex.io/primer/ShExJ\#schema}} and round-trip translators:
\footnote{to JSON: \url{http://shex.io/tools/shex-to-json/}},\footnote{to ShExC: \url{http://shex.io/tools/json-to-shex/}.}

\section{Syntax and semantics of ShEx}
\label{sec:formal}
\subsection{Abstraction of RDF Graphs}
Shape expression schemas allow to constraint both the incoming and
outgoing edges of a node. In order to handle incoming and outgoing
edges uniformly, we consider an abstraction of RDF
graphs. Let $\Prop$ be a set of properties, that in practice
correspond to the set of $\IRI$, and let $\InvProp$ be the set of
inverse properties. An \emph{inverse property} is simply a property
decorated with the hat $\Inv{}$, that is, if $\prop$ is a property,
then $\Inv{\prop}$ is an inverse property. 

A \emph{graph} is defined by a set of nodes $\Nodes$, a set of edges
$\Edges$, and a value function $\val$. The nodes of the graph are
abstract entities. The value function $\val$ associates, with every
node in $\Nodes$, either an $\IRI$, or a literal value, or a special
value $\blank$ that stands for a blank node. Finally, every edge in
$\Edges$ is a triple of the form $(\node, \mathsf{q}, \node')$, where
$\node$ and $\node'$ are two nodes, and $\mathsf{q}$ is a property or
is an inverse property. The nodes $\node$ and $\node'$ are called the
\emph{source} and the \emph{target} nodes of the edge, respectively.

Given a set of triples defining an RDF graph, the above abstraction is
obtained by:
\begin{itemize}
\item For all $\IRI$ $I$ that appears in the subject or
  object position in some triple, there is $\node_{I}$ in
  $\Nodes$ such that $\val(\node_{I}) = I$.
\item For all blank node $B$ that appears in some triple, there is a
  $\node_B$ in $\Nodes$ such that $\val(\node_B) = \blank$.
\item For all triple $(\mathsf{subj}, \prop, \mathsf{obj})$, the set
  $\Edges$ contains the \emph{forward edge} $(\node_{\mathsf{subj}}, \prop,
  \node_{\mathsf{obj}})$ and the \emph{inverse edge} $(\node_{\mathsf{obj}}, \Inv{\prop},
  \node_{\mathsf{subj}})$.
\end{itemize}

Given a $\node$ in a graph, its \emph{neighbourhood} is the set of its
adjacent edges, that is, the edges of the form $(\node, \prop,
\node')$ and $(\node', \Inv{\prop}, \node)$ that belong to
$\Edges$. We denote this set by $\neigh(\node)$. For instance, on the Running Example, the neighbourhood of the node $\node_{\mathsf{ex:issue1}}$ is composed of the edges (that correspond to the triples) $\edge_1,\ldots,\edge_5$  and the inverse of $\edge_6$. A \emph{set of
  neighbourhood edges} is a subset of the neighbourhood of some node.

\subsection{Syntax of ShEx Schemas}
In the abstract syntax presented here, we omit the curly braces that
enclose shape definitions, and the {\sf @} sign preceding referenced
shape names.  A shape expression schema (ShEx schema) defines a set of
shape labels and their associated shape definitions.
\begin{align*}
  \ShapeExpressionSchema \coloncolonequals &
  (\ShapeLabel\  \ShapeDefinition)^+ \\
  \ShapeDefinition \coloncolonequals{}& \text{\sf{}'$\CLOSED$'}?\ \text{\sf{}'$\Inv{\CLOSED}$\sf{}'}?  \\
  & (\mathsf{EXTRA}\ \ExtraPropSet)?\ \ShapeExpr\\
  \ExtraPropSet \colonequals & \text{a set of properties or inverse properties}
\end{align*}
A shape definition is composed of a shape expression ({\sf ShapeExpr}),
possibly preceded by the optional modifiers $\CLOSED$ (forward closed), $\Inv\CLOSED$ (inverse closed)
or $\mathsf{EXTRA}$, the latter taking as parameter a set of
properties or inverse properties.

A shape expression specifies the actual constraints on the
neighbourhood, and is defined by the following syntax rule:
\begin{align*}
\ShapeExpr \coloncolonequals{} & \EmptyShape \mid{} \TripleConstraint \mid{} \SomeOfShape\\
\mid{} & \GroupShape \mid{} \RepetitionShape  
\end{align*}
The \emph{empty shape} ($\EmptyShape$) imposes no constraints. An empty (forward) $\CLOSED$ shape can only be satisfied by a node without outgoing edges, while an empty shape that is not (forward) $\CLOSED$ can have any outgoing edges.

The basic component of a shape expression is a \emph{triple constraint}.
\begin{align*}
&\TripleConstraint \coloncolonequals (\Property \mid \InvProperty)\ \ValueClass\\
&\Property \coloncolonequals \text{an $\IRI$}\\
&\InvProperty \coloncolonequals \Inv{\Property}
\end{align*}
A \emph{triple constraint} is either a forward constraint of the form
$\prop\ \mathsf{K}$, or an inverse constraint of the form
$\Inv{\prop}\ \mathsf{K}$.%
  A triple constraint is satisfied by a single edge adjacent to the
  focus node which node opposite to the focus node satisfies the
  constraint defined by $\mathsf{K}$, and which label is the property
  of the triple constraint. Examples of triple constraints are {\sf
    foaf:name xsd:string} and {\sf is:re\-por\-tedBy @<UserShape>
    \AND\ @<ClientShape>}. If the focus node is $\node$, then an edge
  $(\node, \textsf{foaf:name}, \node')$ satisfies the triple
  constraint {\sf foaf:name xsd:string} if $\val(\node')$ is a
  string. %
The form of the $\ValueClass$ constraint $\mathsf{K}$ is described
below.

A \emph{some-of shape} expression defines a disjunctive
constraint.
\begin{align*}
  \SomeOfShape \coloncolonequals \ShapeExpr\ (\text{\sf{}'|'}\ \ShapeExpr)^*
\end{align*}
It is composed by one or more sub-expressions, separated by the $|$
sign. A some-of expression is satisfied by a set of neighbourhood
edges if at least one of its sub-expressions is satisfied. An example of a
some-of shape 
is {\sf
  foaf:givenName xsd:string | foaf:firstName xsd:string}.

A \emph{group shape} expression is composed by one or more sub-expressions, separated by the
comma sign.
\begin{align*}
\GroupShape \coloncolonequals \ShapeExpr\ (\text{\sf{}','}\ \ShapeExpr)^* 
\end{align*}
It is satisfied by a set of neighbourhood edges if the set
can be split into as many disjoint subsets, and each of these
components satisfies the corresponding sub-expression. An example of a
group shape is {\sf (foaf:givenName xsd:string | foaf:firstName
  xsd:string), foaf:mbox xsd:string}.

Finally, a \emph{repetition shape} expression is composed by an inner
sub-expression, and an allowed number of repetitions specified by a
possibly unbounded interval of natural values. 
\begin{align*}
& \RepetitionShape \coloncolonequals \ShapeExpr\ \Cardinality \\
& \Cardinality \coloncolonequals 
\text{\sf{}'['}\ 
\MinCardinality\  
\text{\sf{}';'}\ 
\MaxCardinality\ 
\text{\sf{}']'}\\
&\MinCardinality \coloncolonequals \text{\sf a natural number}\\
&\MaxCardinality \coloncolonequals 
\text{\sf a natural number} \mid \text{'$\unbound$'}
\end{align*}
Its satisfiability is similar to that of a group expression. A set of
neighbourhood edges satisfies a group expression if it can be split in
$\mathsf{m}$ disjoint subsets, and each of those satisfies the
sub-expression, where $\mathsf{m}$ must belong to the interval of
allowed repetitions. An example of a repetition shape is {\sf
  is:reproducedBy @<ProgrammerShape>+}, where {\sf +} is a short for the
interval $[1,\unbound{}]$.

Hera are the constraints defined by {\sf ValueClass}:
\begin{align*}
&\ValueClass \coloncolonequals \AtomicConstr\  (\text{\sf{}'$\AND$'}\ \AtomicConstr)^* \\
&\AtomicConstr \coloncolonequals \ValueSet \mid{} \ShapeConstr\\
&\ValueSet \coloncolonequals \text{\sf set whose elements are literals, $\IRI$s, or $\blank$}\\ 
&\ShapeConstr \coloncolonequals \ShapeLabel \mid{} \text{\sf{}'\Neg\sf{}'} \ShapeLabel\\
&\ShapeLabel \coloncolonequals \text{\sf an identifier}
\end{align*}
It is a conjunction of sets of values ($\ValueSet$), and of shape
constraints ($\ShapeConstr$). A $\ValueSet$ can contain $\IRI$s,
literal values, or the special constant $\blank$ for a blank node. In
practice, it can be given by listing the values (such as {\sf
  (is:senior is:junior)}), by an XSD value type (such as {\sf
  xsd:string, xsd:integer}), by a node kind specification (for example
IRI, blank, literal, non-literal), by a regular expression defining the
allowed IRIs, etc. A neighbour node of the focus node satisfies a
$\ValueSet$ constraint if its value belongs to the corresponding set.
A shape constraint $\ShapeConstr$ is either a shape label, or a
\emph{negated shape label}, that is a shape label preceded by the bang
($\Neg$) symbol. A neighbour node of the focus node satisfies a non
negated $\ShapeLabel$ if its neighbourhood satisfies the shape
definition that corresponds to that shape label. A negated
$\ShapeLabel$ is satisfied by a node if its neighbourhood does not
satisfy the corresponding shape definition.

\subsection{Well-defined Schemas}
We assume a fixed ShEx schema $\mathsf{Sch}$ whose set of
labels is $\Shapes$. For a shape label $\mathsf{S}$, we denote by
$\definition(\mathsf{S})$ its shape definition in $\mathsf{Sch}$, and
by $\expr(\mathsf{S})$ the shape expression within the definition of
$\mathsf{S}$.

The \emph{dependency graph} of the schema $\mathsf{Sch}$ describes how
shape labels refer to each other in shape definitions. More formally,
the dependency graph of $\mathsf{Sch}$ is an oriented graph whose nodes
are the elements of $\Shapes$, and that has an edge from $\mathsf{S}$
to $\mathsf{T}$ if the shape label $\mathsf{T}$ appears in some triple
constraint in $\expr(\mathsf{S})$.

Let $\mathsf{S}$ and $\mathsf{T}$ be shape labels. We say that
$\mathsf{T}$ \emph{appears negated} in $\definition(\mathsf{S})$ if
either $\Neg\mathsf{T}$ appears in $\definition(\mathsf{S})$, or there
is some triple constraint $\mathsf{q}\ X_1\ \AND \ldots \AND\ X_k$ in
$\expr(\mathsf{S})$ such that $\mathsf{q}$ is an extra property in
$\definition(\mathsf{S})$, and $\mathsf{T}$ is a shape label among
$X_1, \ldots, X_k$. For instance, in the Running Example, 
the shape labels {\sf <TesterShape>} and {\sf
  <ProgrammerShape>} appear negated in the definition of {\sf <IssueShape>}
because {\sf is:reportedBy} is an extra property. We denote by
$\negatedShapes(\mathsf{S})$ the set of shape labels that appear
negated in $\definition(\mathsf{S})$, and 
$\negatedShapes(\mathsf{Sch})$ is the union of
$\negatedShapes(\mathsf{S})$ for all $\mathsf{S}$ in $\Shapes$.

The following syntactic restriction is imposed in order to guarantee
well-defined semantics for ShEx schemas in presence of
recursion and negation. It requires that shape labels that appear
negated do not lead to cyclic dependencies between shapes. 
From now on, we assume that all schemas are well-defined.
\begin{DEF}[Well-defined schema]
  \label{def:well-defined-shex}
  A shape expression schema $\mathsf{Sch}$ is \emph{well-defined} if
  for every shape labels $\mathsf{S}, \mathsf{T}$, if $\mathsf{T}$ is
  in $\negatedShapes(\mathsf{S})$, then the sub-graph accessible from
  $\mathsf{T}$ in the dependency graph of $\mathsf{Sch}$ is a direct
  acyclic graph.
\end{DEF}

\subsection{Declarative Semantics of ShEx}

\paragraph{Locally Satisfying a Shape Definition} For every shape
definition, we need to refer to the occurrences of its triple
constraints. Therefore, we are going to use
$\mathsf{C}_1,\ldots,\mathsf{C}_k$ as unique names for the triple
constraints that appear in a shape definition. The shape definition to
which the $\mathsf{C}_i$ belong will be clear from the context. Note
that if the same triple constraint appears twice (i.e. same property
and same value class), the two occurrences are distinguished and
correspond to different $\mathsf{C}_i$'s. To say it differently, any
of the $\mathsf{C}_i$ corresponds to a $\TripleConstraint$-position in
the abstract syntax tree of a shape definition.

For every shape definition we identify a set of triple consumers, that
correspond either to some triple constraint, or to an extra property,
or to the unconstrained properties of open (i.e. not closed) shape
definitions. Intuitively, a $\node$ locally satisfies a shape
definition if all edge from $\neigh(\node)$ can be consumed by
one of the triple consumers of that shape definition, in a way that
satisfies its shape expression.

Formally, let $\ShDef$ be a shape definition (fixed in the sequel),
and let $\mathsf{C}_1, \ldots, \mathsf{C}_k$ be the set of its triple
constraints. The set of \emph{triple consumers} of $\ShDef$ consists
of:
\begin{itemize}
\item $\TCons_{\mathsf{C}_i}$ for all triple constraint $\mathsf{C}_i$ in $\ShDef$;
\item $\TCons_{\mathsf{q},\extra}$ for all extra property $\mathsf{q}$ in
  $\ShDef$;
\item $\TCons_\open$ which is a special constant.
\end{itemize}

We say that an edge $\edge = (\node, \mathsf{q}, \node')$
\emph{matches} a triple consumer $\TCons$ if either $\TCons =
\TCons_{\mathsf{q}, \extra}$, or $\TCons =
\TCons_{\mathsf{C}_i}$ with $\mathsf{C}_i = \mathsf{q}\ X_1
\AND{} \ldots \AND{} X_k$ and for all $1 \le j \le k$ such that $X_j$ is a
value set, it holds that $\val(\node') \in X_j$.

\begin{DEF*}[Local witness]
  \label{def:local-witness}
  Let $\ShDef$ be a shape definition, $\Consumers$ be the its set of
  triple consumers, and let $\node$ be a node.  Let
  $\witness : \neigh(\node) \rightarrow \Consumers$ be a total mapping
  that maps a triple consumer with every edge in $\neigh(\node)$. We
  say that the mapping $\witness$ is a \emph{local witness for the
    fact that $\neigh(\node)$ satisfies $\ShDef$}, written $\witness,
  \neigh(\node) \vdash \ShDef$, iff:
  \begin{itemize}
  \item For all edge $\edge = (\node, \mathsf{q}, \node')$ in
    $\neigh(\node)$:\\
    $-$ $\edge$ matches $\witness(\edge)$ whenever $\witness(\edge)$ is of the form $\TCons_{\mathsf{q}, \extra}$ or $\TCons_{\mathsf{C}_i}$;\\
    $-$ if $\witness(\edge) = \TCons_{\mathsf{q}, \extra}$, then there
    is no $\TCons_{\mathsf{C}_i}$ in $\Consumers$ s.t. and all the conjuncts in
    $\mathsf{C}_i$ are value sets and $\edge$ matches $\TCons_{\mathsf{C}_i}$ ;\\
    $-$ if $\witness(\edge) = \TCons_{\open}$, then $\Consumers$ does
    not contain any triple consumer of the form
    $\TCons_{\mathsf{q},\extra}$ or $\TCons_{\mathsf{C}_i}$ with
    $\mathsf{C}_i$ of the form $\mathsf{q}\ \mathsf{K}$ (i.e. having
    the same property $\mathsf{q}$).
  \item If $\ShDef$ is forward closed (respectively, inverse closed),
    then there is no forward edge (respectively, inverse edge) that is
    mapped with $\TCons_\open$ by $\witness$.
  \item Let $\Neigh_\expr$ be the set of edges $\edge$ from
    $\neigh(\node)$ such that $\witness(\edge)=\TCons_{\mathsf{C}_i}$
    for some triple constraint $\mathsf{C}_i$ in $\ShDef$, then it holds
    that $\witness, \Neigh_\expr \vdash \expr(\ShDef)$.
  \end{itemize}
  
  For a set of neighbourhood edges $\Neigh$, a shape expression $\Expr$,
  and a mapping $\witness: \Neigh \to \Consumers$, we say that
  $\witness$ is a \emph{local witness for the fact that $\Neigh$
    satisfies $\Expr$}, written $\witness, \Neigh \vdash \Expr$, iff:
  \begin{itemize}
  \item $\Expr$ is the empty shape, and $\Neigh = \emptyset$;
  \item $\Expr = \mathsf{C}_i$ is a triple constraint, $\Neigh =
    \{\edge\}$ is a singleton set, and $\witness(\edge) = \mathsf{C}_i$.
  \item $\Expr = \Expr_1 \mid\ldots\mid \Expr_k$ is a some-of shape, and
    $\witness, \Neigh \vdash \Expr_i$ for some $1\leq i \leq k$.
  \item $\Expr = \Expr_1,\ldots,\Expr_k$ is a group shape, and for all
    $1\leq i\leq k$, denote by $\Neigh_i$ the subset of $\Neigh$ that
    contains $\edge$ iff $\witness(\edge) =
    \TCons_{\mathsf{C}_j}$ for some triple constraint
    $\mathsf{C}_j$ in $\Expr_i$. Then it holds that $\Neigh =
    \Neigh_{1}\cup\ldots\cup\Neigh_{k}$, and $\witness, \Neigh_{i} \vdash
    \Expr_i$ for all $1 \leq i \leq k$.
  \item $\Expr\mathsf{[min;max]}$ is a repetition shape, and there
    exists $\mathsf{m}$ that belongs to the interval $\mathsf{[min;max]}$ such
    that one can partition $\Neigh$ in $m$ disjoint sets
    $\Neigh_1,\ldots,\Neigh_m$ whose union is $\Neigh$, and such that
    $\witness, \Neigh_i \vdash \Expr$ for all $1\leq i \leq m$. \qed
  \end{itemize} 
\end{DEF*}

\begin{example}
\label{ex:local-witness}
With the schema and data from the Running Example, and with $\TCons_0$ as short for $\TCons_{\textsf{is:reproducedBy},\extra}$, we have that 
$\edge_1$ matches $\TCons_{\mathsf{C}_1}$, $\edge_5$
matches $\TCons_\open$, (the inverse of) $\edge_6$ matches $\TCons_{\mathsf{C}_4}$, and $\edge_j$ matches $\TCons_0,
\TCons_{\mathsf{C}_2}, \TCons_{\mathsf{C}_3}$, for $j = 2,3,4$. The mapping $\witness$ defined here after is a local witness for the fact that {\sf ex:issue} satisfies the definition of {\sf <IssueShape>}. 
  \begin{multline}
    \label{eq:witness}
    \edge_1 \mapsto \TCons_{\mathsf{C}_1}, \edge_5 \mapsto \TCons_\open, \edge_6 \mapsto \TCons_{\mathsf{C}_4}, \\
\edge_2 \mapsto \TCons_{\mathsf{C}_2},
    \edge_3 \mapsto \TCons_{\mathsf{C}_3}, \edge_4 \mapsto \TCons_{\mathsf{C}_3}
  \end{multline}
\end{example}

\paragraph{Global Typing Witness}
The shape labels that appear in triple constraints allow to propagate
constraints beyond the immediate neighbourhood of a node. Thus, the
validity of a graph with respect to a ShEx schema is a
global property on the graph, and is captured by the notion of a global
typing witness to be defined shortly.

As previously, we consider a ShEx schema $\mathsf{Sch}$
whose set of shape labels is $\Shapes$. Let $\NegatedShapes$ be the
set of shape labels of the form $\Neg{\mathsf{S}}$ where $\mathsf{S}$
is in $\negatedShapes(\mathsf{Sch})$. That is, $\NegatedShapes$
contains all the shapes that appear negated in $\mathsf{Sch}$,
decorated by a leading $\Neg{}$ sign.

A \emph{typing} of a graph $\mathsf{G}$ by $\mathsf{Sch}$ is a set
$\typing \subseteq \Nodes_\mathsf{G} \times (\Shapes \cup
\NegatedShapes)$ of couples of the form $(\node, \mathsf{S})$ or
$(\node, \Neg\mathsf{S})$, and such that there is no $\node$ and no shape
label $\mathsf{S}$ for which both $(\node, \mathsf{S})$ and $(\node,
\Neg{\mathsf{S}})$ belong to $\typing$. 
Let $\witness$ be a local witness, then the \emph{propagation} of
$\witness$ is the typing $\mathsf{propagation}_{\mathsf{witness}}$
that contains precisely the couples $(\node', X)$ for all edge
$(\node, \mathsf{q}, \node')$ in the domain of $\witness$ such that
$\witness(\node, \mathsf{q}, \node') = \TCons_{\mathsf{C}_i}$
corresponds to a triple constraint $\mathsf{C_i}$, and $X$ is a shape
label or negated shape label that appears as a conjunct in
$\mathsf{C}_i$.

\begin{DEF}[Global typing witness]
  \label{def:global-typing-witness}
  A \emph{global typing witness} for a graph $\mathsf{G}$ by a
  schema $\mathsf{Sch}$ is a couple $\typing,\lw$, where $\typing$ is
  a typing of $\mathsf{G}$ by $\mathsf{Sch}$, and $\lw$ is a total map
  from $\typing \cap (\Nodes \times \Shapes)$ to local witnesses,
  s.t. for all $(\node, \mathsf{S})$ in $\typing$, it holds that
  $\lw(\node, \mathsf{S})$ is a local witness for the fact that
  $\node$ satisfies $\definition(\mathsf{S})$. Additionally:
\begin{description}
\item [gtw-sat] The constraints required by every (non negated) shape
  label are propagated. That is, for all $(\node, \mathsf{S}) \in
  \typing$, it holds $\mathsf{propagation}_{\lw(\node, \mathsf{S})}
  \subseteq \typing$;
\item[gtw-neg] The negated shape labels cannot be satisfied, that is, if
  $(\node, \Neg{\mathsf{S}}) \in\typing$, then there does not exist a
  global typing witness $\typing', \lw'$ s.t. $(\node, \mathsf{S})
  \in\typing'$.
\item[gtw-extra] The edges consumed by extra consumers do not satisfy
  the corresponding triple constraints. That is, for all shape label
  $\mathsf{S}$ and all $\edge = (\node, \mathsf{q}, \node')$
  s.t. $\lw(\node, \mathsf{S})(\edge) = \TCons_{\mathsf{q}, \extra}$,
  and for all triple consumer $\TCons_{\mathsf{C}_i}$ that corresponds
  to a triple constraint $\mathsf{C}_i = \mathsf{q}\ X_1 \AND \ldots
  \AND X_k$ in $\definition(\mathsf{S})$, there is a $1 \le j \le k$
  such that
  \begin{description}
  \item[gtw-extra-value-set] either $X_j$ is a value set and\\
    $\val(\node') \not\in X_j$,
  \item[gtw-extra-shape-constr] or $X_j$ is a shape constraint, and
    there exists a global typing witness $\typing'', \lw''$ such that
    $(\node', Y_j) \in \typing''$, where $Y_j = \Neg X_j$ if $X_j$ is
    a shape label, and $Y_j = \mathsf{T}$ if $X_j = \Neg\mathsf{T}$ is
    a negated shape label. \qed
  \end{description}
\end{description}
\end{DEF}

\begin{example}
  \label{ex:global-typing-witness}
  With the shape and data from the Running Example, and with $\witness$ from Example~\ref{ex:local-witness}, there is no global typing witness that includes $\mathsf{propagation}_\witness$, because ({\sf ex:emin}, {\sf <ProgrammerShape>}) is in $\mathsf{propagation}_\witness$ and {\sf ex:emin} does not satisfy {\sf <ProgrammerShape>}. 
Let $\typing$ be the typing that with every subject node of the Running Example associates the shape labels listed in the comment on the right of that node in the example data,  e.g. $\typing$ contains ({\sf ex:ren}, {\sf <TesterShape>}), ({\sf ex:fatima}, {\sf <UserShape}>), ({\sf ex:fatima}, {\sf <ClientShape}>), etc. 
Then $\typing$ corresponds to a global typing witness, with $\lw(\node_{\textsf{ex:issue1}}, \textsf{<IssueShape>}) = \witness'$, with $\witness'$ identical to $\witness$ except for $\edge_4 \mapsto \TCons_0$.~$\square$
\end{example}

The definition of a global typing witness is recursive: in the \textbf{gtw-neg} we require
to ensure that some $\typing',\lw'$ is not a global typing witness
in order to ensure that $\typing,\lw$ is a global typing witness. It
what follows, we give some fundamental properties of global typing
witnesses that allow to show that the definition is well-founded.

\begin{theorem}
  \label{thm:certain-typing}
  For all graph $\mathsf{G}$ and all shape expression sche\-ma
  $\mathsf{Sch}$, there exists a global typing witness
  $\typingcert_{\mathsf{G}, \mathsf{Sch}},\lwcert_{\mathsf{G},
    \mathsf{Sch}}$ such that for all node $\node$ in $\mathsf{G}$ and all
  shape label $\mathsf{S} \in \negatedShapes(\mathsf{Sch})$, either
  $(\node,\mathsf{S})$ or $(\node, \Neg\mathsf{S})$ belongs to  $\typingcert_{\mathsf{G}, \mathsf{Sch}}$. 
\end{theorem}
\begin{proof} (Sketch.) We sketch the proof using an example, and show
  how $\typingcert_{\mathsf{G}, \mathsf{Sch}},\lwcert_{\mathsf{G},
    \mathsf{Sch}}$ can be effectively computed, starting with an empty
  typing. The proof is based on the well-founded criterion of
  schemas. Suppose that $\negatedShapes(\mathsf{Sch}) =
  \{\mathsf{S}_1, \mathsf{S}_2, \mathsf{S}_3, \mathsf{S}_4\}$ and the
  dependency among them is $\mathsf{S}_1 \rightarrow \mathsf{S}_2$,
  $\mathsf{S}_1 \rightarrow \mathsf{S}_3$, $\mathsf{S}_2 \rightarrow
  \mathsf{S}_4$, $\mathsf{S}_3 \rightarrow \mathsf{S}_4$ (where
  $\mathsf{S}_1 \rightarrow \mathsf{S}_2$ means that $\mathsf{S}_2$
  appears in the definition of $\mathsf{S}_1$). We start by typing
  with $\mathsf{S}_4$, as follows. For all $\node$, if there exists a
  local witness $\witness$ for the fact that $\node$ satisfies
  $\mathsf{S}_4$, then we add $(\node, \mathsf{S}_4)$ to
  $\typingcert_{\mathsf{G}, \mathsf{Sch}}$ and we set
  $\lwcert_{\mathsf{G}, \mathsf{Sch}}(\node, \mathsf{S}_4) =
  \witness$. Otherwise, we add $(\node, \Neg\mathsf{S}_4)$ to
  $\typingcert_{\mathsf{G}, \mathsf{Sch}}$. After this first step,
  $\typingcert_{\mathsf{G}, \mathsf{Sch}},\lwcert_{\mathsf{G},
    \mathsf{Sch}}$ is a global typing witness. Indeed, it satisfies
  \textbf{gtw-sat} by definition. For \textbf{gtw-neg}, suppose by
  contradiction that $(\node, \Neg\mathsf{S}_4) \in
  \typingcert_{\mathsf{G}, \mathsf{Sch}}$ and there is a global typing
  witness $\typing',\lw'$ s.t. $(\node, \mathsf{S}_4) \in \typing'$,
  then $\lw'(\node, \mathsf{S}_4)$ is a local witness for the fact
  that $\node$ satisfies $\mathsf{S}_4$: this is a contradiction. For
  \textbf{gtw-extra}, the proof goes again by contradiction: if $\lwcert_{\mathsf{G},
    \mathsf{Sch}}(\node, \mathsf{S}_4)(\node, \mathsf{q}, \node') =
  \TCons_{\mathsf{q}, \extra}$ and there is $\TCons_{\mathsf{C}_i}$
  consumer in the definition of $\mathsf{S}_4$ s.t. $\node'$ is in
  $\val(X_j)$ for all $X_j$ conjunct in $\mathsf{C}_i$, then
  $\lwcert_{\mathsf{G}, \mathsf{Sch}}(\node, \mathsf{S}_4)$ does not
  satisfy the definition of local witness: contradiction.
  
  We next associate the shapes $\mathsf{S}_2$ and
  $\mathsf{S}_3$ to all nodes. We can do it in any order, because they
  both only depend on $\mathsf{S}_4$. We illustrate taking
  $\mathsf{S}_2$. Let $\typing^{\mathsf{prev}}, \lw^{\mathsf{prev}}$
  be $\typingcert_{\mathsf{G}, \mathsf{Sch}}, \lw_{\mathsf{G},
    \mathsf{Sch}} $ as obtained during the previous step (i.e.,
  $\typing^{\mathsf{prev}}$ uses only $\mathsf{S}_4$ and
  $\Neg\mathsf{S}$ as shapes). For all $\node$, if there exists a local
  witness $\witness$ for the fact that $\node$ satisfies
  $\mathsf{S}_2$ and s.t. $\mathsf{propagation}_\witness \subseteq
  \typing^{\mathsf{prev}}$, we add $(\node, \mathsf{S}_2)$ to
  $\typingcert_{\mathsf{G}, \mathsf{Sch}}$, otherwise we add $(\node,
  \Neg\mathsf{S}_2)$ to $\typingcert_{\mathsf{G},
    \mathsf{Sch}}$. Using similar arguments as for the previous step,
  and that $\typing^{\mathsf{prev}},\lw^{\mathsf{prev}}$ is a
  global typing witness, we show that the new
  $\typingcert_{\mathsf{G}, \mathsf{Sch}}, \lwcert_{\mathsf{G},
    \mathsf{Sch}}$ is a global typing witness too.

  This process is repeated until all negated shape labels have been
  processed, and following the ordering induced by the (acyclic)
  dependency graph. For instance, $\mathsf{S}_1$ can be added only
  after $\mathsf{S}_2$ and $\mathsf{S}_3$ are both added.
\end{proof}

The following Corollary~\ref{cor:certain-typing} establishes that all global typing witness
agrees with $\typingcert_{\mathsf{G},
  \mathsf{Sch}},\lwcert_{\mathsf{G}, \mathsf{Sch}}$ on the shape labels
that appear negated in $\mathsf{Sch}$. It follows from the proof of
Theorem~\ref{thm:certain-typing} and from
Definition~\ref{def:global-typing-witness}. This allows us to give an equivalent, non recursive definition for a
global typing witness in Lemma~\ref{lem:def-global-typing-witness}. 

\begin{corollary}
  \label{cor:certain-typing}
  If $\typing,\lw$ is a global typing witness for the graph
  $\mathsf{G}$ by the schema $\mathsf{Sch}$, then for all $\mathsf{S}
  \in \negatedShapes(\mathsf{Sch})$, and for all node $\node$ in
  $\mathsf{G}$, if $(\node, \mathsf{S}) \in \typing$, then $(\node,
  \mathsf{S}) \in \typingcert_{\mathsf{G}, \mathsf{Sch}}$, and if $(\node,
  \Neg\mathsf{S}) \in \typing$, then $(\node, \Neg\mathsf{S}) \in
  \typingcert_{\mathsf{G}, \mathsf{Sch}}$.
\end{corollary}

\begin{lemma}
  \label{lem:def-global-typing-witness}
  In Definition~\ref{def:global-typing-witness}, the \textbf{\upshape
    gtw-neg} and \textbf{\upshape gtw-extra-shape-constr} conditions can
  be replaced by the following weaker conditions \textbf{\upshape
    gtw-neg${}'$} and \textbf{\upshape gtw-extra-shape-constr${}'$}
  respectively, while leaving the underlying notion of global typing
  witness unchanged:
  \begin{description}
  \item[gtw-neg${}'$] $(\node, \Neg\mathsf{S}) \in \typing$ only if
    $(\node, \Neg\mathsf{S}) \in \typingcert_{\mathsf{G},
      \mathsf{Sch}}$.
  \item[gtw-extra-shape-constr${}'$] or $X_j$ is a shape constraint, and
    $(\node', Y_j) \in \typingcert_{\mathsf{G}, \mathsf{Sch}}$, where
    $Y_j = \Neg X_j$ if $X_j$ is a shape label, and $Y_j = \mathsf{T}$
    if $X_j = \Neg\mathsf{T}$ is a negated shape label.
  \end{description}
\end{lemma}


\section{Validation Algorithm }
\label{sec:validation}

The fundamental question in validation is ``does X satisfy Y?'', e.g. ``does {\sf ex:issue1} satisfy {\sf <IssueShape>}?''.
Following is a validation algorithm which, given an initial typing $\typing_0$ that contains typing requirements such as {\sf (ex:issue1, <IssueShape>)}, constructs a \emph{global typing witness} $\typing,\lw$ that includes $\typing_0$ if one exists, and raises a validation error otherwise.
Throughout the section we consider a fixed
graph $\mathsf{G}$ with nodes $\Nodes$, edges $\Edges$ and a value
function $\val$, a fixed shape expression schema $\mathsf{Sch}$ over a
set of shape labels $\Shapes$, and a fixed (partial) initial typing
$\typing_0$ of $\mathsf{G}$ by $\mathsf{Sch}$. We start by presenting
a high-level version of the algorithm, and then discuss some
implementation and optimization aspects.

\subsection{Data Structures}
\label{sec:validation-data-structures}

The flooding algorithm produces a global typing witness. It proceeds
by making \emph{typing hypotheses}, that are associations of a node
and a shape label. The algorithm tries to witness that each such
hypothesis is satisfied, or otherwise removes it, until a global
typing witness is obtained. The satisfaction of a typing hypothesis
might \emph{require} other typing hypotheses to be satisfied. During
its computation, the algorithm maintains a structure called
\emph{typing witness under construction} denoted $\TUC$, which
contains the current hypotheses, together with additional data
useful for the computation. More precisely $\TUC = (\typinghyp,
\lwhyp, \requires, \toCheck)$, where $\typinghyp,\lwhyp$ is the global
typing witness under construction, $\requires \subseteq \typinghyp
\times \typinghyp$ is a binary relation on hypotheses, and $\toCheck$
is a subset of $\typinghyp \cap (\Nodes \times \Shapes)$ of not yet
verified hypotheses. We suppose that $\typingcert_{\mathsf{G},
  \mathsf{Sch}}, \lwcert_{\mathsf{G}, \mathsf{Sch}}$ is given, denoted
$\typingcert, \lwcert$ for short. Finally, the algorithm uses a global
map $\ULW$ that, with every node $\node$ and shape label $\mathsf{S}$,
associates the set of local witnesses that can potentially be used for
proving that $\node$ satisfies $\mathsf{S}$. A local witness
$\witness$ is removed from $\ULW(\node, \mathsf{S})$ when we know that
it cannot be used in any global typing witness, that is, when there
does not exist a global typing witness $\typing,\lw$ s.t. $(\node,
\mathsf{S})$ belongs to $\typing$ and $\mathsf{propagation}_\witness$
is included in $\typing$.

\subsection{The Flooding Algorithm}
\label{sec:validation-flooding-algo}
The flooding algorithm is presented in
Algorithm~\ref{algo:flooding}. It starts by checking whether the
initial typing $\typing_0$ is compatible with the certain types
$\typingcert$, and if not, it signals a validation error. Otherwise,
the typing under construction is initialized so that it contains
$\typing_0$ as initial hypotheses, all the initial hypotheses
involving a non-negated shape label are added to $\toCheck$, whereas
the $\lwhyp$ and $\requires$ components are empty
(lines~\ref{lalgo:flooding-init-start}
to~\ref{lalgo:flooding-init-end}). The function
$\mathsf{checkCompatible}$ takes as parameters two typings and returns
false iff there are a node $\node'$ and a shape label $\mathsf{T}$
s.t. one of the typings contains $(\node', \mathsf{T})$ and the other
contains $(\node', \Neg\mathsf{T})$.

\begin{algorithm}
  \caption{$\mathsf{FloodingValidation}$\label{algo:flooding}}

  \KwIn {$\mathsf{Sch}$ a shape expression schema over $\Shapes$,\\
    \hspace{2em} $\mathsf{G} = (\Nodes,\Edges, \val)$ a graph,\\
    \hspace{2em} $\typing_0$ a pre-typing of $\mathsf{G}$ by $\mathsf{Sch}$,\\
    \hspace{2em} $\ULW$ a map from $\Nodes \times \Shapes$ to sets of  local witnesses\\
    \hspace{2em} $\typingcert, \lwcert$ a global typing witness}
  \KwOut{$\typing,\lw$ a global typing witness}

  \BlankLine
  \If {\KwNot $\mathsf{checkCompatible}(\typing_0, \typingcert)$} {
    \Return {\sf VALIDATION\_ERROR}
  }
 \BlankLine

 \KwLet $\TUC \colonequals (\typinghyp, \lwhyp, \requires, \toCheck)$, where  \label{lalgo:flooding-init-start}\\
 \Indp $\typinghyp = \typing_0$, $\lwhyp = \emptyset$,  $\requires = \emptyset$, \\
 and $\toCheck = \typing_0 \cap (\Nodes \times \Shapes)$  \label{lalgo:flooding-init-end}\\
 \Indm

  \BlankLine
  \While {$\toCheck \neq \emptyset$ \label{lalgo:flooding-main-loop}}{ 
    $(\node, \mathsf{S}) \colonequals $ remove from $\toCheck$ 
    \BlankLine

    \If {$(\node, \mathsf{S}) \in \typingcert$ \label{lalgo:flooding-if-certain}} {
      \KwCont
    }
    \BlankLine
    
    \ElseIf  {$\ULW(\node, \mathsf{S}) = \emptyset$ \label{lalgo:flooding-if-unsatisfied}} {
      $\mathsf{Backtracking}(\node, \mathsf{S}, \TUC)$
    }
    \BlankLine
    
    \Else(\label{lalgo:flooding-if-shape-name}) {
      $\witness \colonequals$  get first from $\ULW(\node, \mathsf{S})$ \label{lalgo:flooding-get-first}\;
      
      \If {$\mathsf{checkGtwExtra}(\mathsf{Sch}, \mathsf{S}, \witness, \typingcert)$ \KwAnd $ \mathsf{checkCompatible}(\mathsf{propagation}_\witness, \typingcert)$} {
        \BlankLine

        $\lwhyp(\node, \mathsf{S}) = \witness$ \label{lalgo:flooding-set-lwcert}\;
        \ForEach {$(\node', X)$ {\upshape in} $\mathsf{propagation}_\witness$ \label{lalgo:flooding-propagate-start}} {
          \If {$(\node', X) \not\in \typinghyp$} {
            add $(\node', X)$ to $\typinghyp$\;
            \If {$X \in \Shapes$} {
              add $(\node', X)$ to $\toCheck$\;
            }
          }
          \If {$X \in \Shapes$} {
            add $((\node, \mathsf{S})(\node', X))$ to $\requires$ \label{lalgo:flooding-propagate-end}\; 
          }
        }
      }
      \Else  {\label{lalgo:extra-violation}
        remove $\witness$ from $\ULW(\node, \mathsf{S})$\;
        add $(\node, \mathsf{S})$ to $\toCheck$\;
      }
    }
  }
  \If {$\typing_0 \subseteq \typinghyp$ \label{lalgo:flooding-after-main-loop}} {
    \ForEach {$(\node, \mathsf{S}) \in \typinghyp \wedge \typingcert$} {
      $\mathsf{copyProof}(\node, \mathsf{S}, \typingcert, \lwcert, \typinghyp, \lwhyp)$
    }
    \Return $\typing = \typinghyp, \lw = \lwhyp$
  } \Else {
    \Return {\sf VALIDATION\_ERROR}
  }
\end{algorithm}

\begin{algorithm}
  \caption{$\mathsf{Backtracking}$\label{algo:backtracking}}
  \KwIn {$(\node, \mathsf{S})$ a hypothesis, \\
    \hspace{3em} $\TUC = (\typinghyp, \lwhyp, \requires, \toCheck)$ \\
    \hspace{3.5em} a global typing witness under construction} 
  \ForEach {$(\node', \mathsf{S}') \in \mathsf{toRemove}(\node, \mathsf{S}, \requires)$ \label{lalgo:backtracking-toremove-loop}} { 
    remove $(\node', \mathsf{S}')$ from $\typinghyp$, $\lwhyp$, and $\toCheck$ \label{lalgo:backtracking-toremove}
  }
\end{algorithm}

The main loop on line~\ref{lalgo:flooding-main-loop} iterates on all
hypotheses in $\toCheck$, until all have been processed. For all
hypothesis $(\node, \mathsf{S})$ in $\toCheck$, we distinguish three
possible cases, that are the conditions on
lines~\ref{lalgo:flooding-if-unsatisfied}, \ref{lalgo:flooding-if-certain}, and~\ref{lalgo:flooding-if-shape-name}.
On line~\ref{lalgo:flooding-if-certain}, the hypothesis $(\node,
\mathsf{S})$ to be checked is known to be verified because
$\mathsf{S}$ is a certain type for $\node$.
Line~\ref{lalgo:flooding-if-unsatisfied} corresponds
to the case when the hypothesis $(\node, \mathsf{S})$ is recognized as
non provable because the set $\ULW(\node, \mathsf{S})$ of unchecked
hypotheses is empty. We need to backtrack, as explained
later on.
Finally, line~\ref{lalgo:flooding-if-shape-name} corresponds to the
case when there exists a local witness $\witness$ that has not been
used yet for verifying whether $\node$ satisfies shape
$\mathsf{S}$. Then we first use the Boolean functions
$\mathsf{checkGtwExtra}$ (see below) to make sure that $\witness$ does
not violate the \textbf{gtw-extra} condition, and
$\mathsf{checkCompatible}$ to make sure that propagating the witness
won't contradict the certain typing. If the check passes, then we set
$\lwhyp(\node, \mathsf{S})$ to $\witness$ (line~\ref{lalgo:flooding-set-lwcert}), meaning that we are going
to look for a valid typing compatible with $\witness$, then we
propagate the further requirements imposed by $\witness$ to the
neighbours of $\node$ (lines~\ref{lalgo:flooding-propagate-start}
to~\ref{lalgo:flooding-propagate-end}). Otherwise
(line~\ref{lalgo:extra-violation}), $\witness$ is removed from $\ULW$,
and $(\node, \mathsf{S})$ is added back to $\toCheck$ for further
checking.
The functions $\mathsf{checkGtwExtra}$ takes as input a schema
$\mathsf{Sch}$, a shape label $\mathsf{S}$, a local witness $\witness$,
and a certain typing $\typingcert$, and returns false iff there is an
edge $\edge = (\node, \mathsf{q}, \node')$ in the domain of $\witness$
s.t. $\witness(\edge) = \TCons_{\mathsf{q}, \extra}$ is an extra
consumer, and there exists a triple consumer $\TCons_{\mathsf{C}_i}$
in the definition of $\mathsf{S}$ (in $\mathsf{Sch}$) such that
$(\node', X) \in \typingcert$, for all shape constraint $X$ that is a
conjunct in $\mathsf{C}_i$.

When all the hypotheses have been processed
(line~\ref{lalgo:flooding-after-main-loop}), $\typinghyp,\lwhyp$ is
(almost) a global typing witness. However, it might not contain
$\typing_0$, because some of the initially required hypotheses have
been disproved (removed during backtracking). In that case, a
validation error is raised. Otherwise, we use the $\mathsf{copyProof}$
procedure to copy the proofs for the certain facts from $\typingcert,
\lwcert$ to $\typinghyp, \lwhyp$. These facts were previously
processed on line~\ref{lalgo:flooding-if-certain}. For all certain
fact $(\node, \mathsf{S})$, the $\mathsf{copyProof}$ procedure first
adds $\mathsf{propagation}_{\lwcert(\node, \mathsf{S})}$ to
$\typinghyp$, and then recursively copies the proofs for the newly
added certain facts. It terminates because the schema is well-founded.

\paragraph{Backtracking}
The backtracking algorithm is described in
Algorithm~\ref{algo:backtracking}. It removes all hypotheses that are
not relevant any more when we find out that a hypothesis $(\node,
\mathsf{S})$ cannot be satisfied. Intuitively, these are the
hypotheses $(\node', \mathsf{S}')$ that required $(\node,
\mathsf{S})$, and all hypotheses that are required (possibly
indirectly) by such $(\node', \mathsf{S}')$. More formally, this is
captured by the set $\mathsf{toRemove}(\node, \mathsf{S}, \requires)
\subseteq \typinghyp$ recursively defined below. That set can be
effectively computed using standard reachability algorithms on graphs.
\begin{itemize}
\item for all $(\node', \mathsf{S}')$ s.t. $requires((\node',
  \mathsf{S'}), (\node, \mathsf{S}))$, $(\node', \mathsf{S}')$ is in
  $\mathsf{toRemove}(\node, \mathsf{S}, \requires)$;
\item if $(\node',\mathsf{S}')$ is such that for all
  $(\node'',\mathsf{S}'')$ we have that
  $\requires((\node'',\mathsf{S}''), (\node',\mathsf{S}'))$ implies
  $(\node'',\mathsf{S}'') \in \mathsf{toRemove}(\node,
  \mathsf{S}, \requires)$, then $(\node',\mathsf{S}')$ is also in
  $\mathsf{toRemove}(\node, \mathsf{S}, \requires)$.
\end{itemize}

Additionally, in the loop on line~\ref{lalgo:backtracking-invalidate},
backtracking invalidates the local witnesses for the hypotheses
$(\node', \mathsf{S}')$ that required $(\node, \mathsf{S})$, and adds
them back for checking.

\newcommand{\potTC}{\mathsf{matchingTC}}
\newcommand{\CandULW}{\mathsf{CandidateULW}}

\paragraph{Computing the Sets of Unchecked Local Witnesses} Let
$\node$ be a node, $\mathsf{S}$ be a shape label which definition
contains the triple constraints $\mathsf{C}_1, \ldots, \mathsf{C}_k$
and has corresponding set $\Consumers$ of triple consumers. We compute
the set $\ULW(\node, \mathsf{S})$ by considering a set of
\emph{candidate} mappings from $\neigh(\node)$ to $\Consumers$, and
keeping those candidates that are actual local witnesses for the fact that
$\node$ satisfies the definition of $\mathsf{S}$. A mapping is a
candidate if it associates with every edge a triple consumer that this
edge matches. More formally, for all $\edge = (\node, \mathsf{q},
\node')$, let $\potTC(\edge) = \{\TCons_\open\}$ if $\mathsf{q}$ does
not appear in any of $\mathsf{C}_i$, neither as an extra property in
$\mathsf{S}$, and  $\potTC(\edge) = \{\TCons_{\mathsf{q}, \extra}\}
\cup \{\TCons \mid \edge \text{ matches } \TCons\} $ otherwise.
Then a candidate mapping is obtained by choosing one triple consumer
among $\potTC(\edge)$ for all $\edge$ in $\neigh(\node)$. For
instance, continuing Example~\ref{ex:local-witness}, a candidate map
from $\neigh(\node_{\textsf{ex:issue1}})$ to the triple consumers of
{\sf <IssueShape>} will always associate $\TCons_\open$ with
$\edge_5$, and will associate one among $\TCons_0$, 
$\TCons_{\mathsf{C}_2}$, and $\TCons_{\mathsf{C}_3}$ with $\edge_4$. We
denote $\CandULW(\node,\mathsf{S})$ the set of candidate mappings, and
it is obtained as the Cartesian product of the sets $\potTC(\edge)$,
for all $\edge$ in $\neigh(\node)$.

Once the set $\CandULW(\node,\mathsf{S})$ of candidate mappings is
constructed, we have to determine which among them are local witnesses
for the fact that $\node$ satisfies the definition of
$\mathsf{S}$. This can be done using algorithms that we proposed in
\cite{jose-rdf-validation} and \cite{icdt2015}. For that,
$\neigh(\node)$ is seen as a bag over the alphabet $\Consumers$, by
replacing every $\edge$ by $\witness(\edge)$. A \emph{bag} is an
unordered collection with possibly repeated symbols. For instance, the
bag that corresponds to the mapping (\ref{eq:witness}) from
Example~\ref{ex:local-witness} is
$\ml \TCons_{\mathsf{C}_1},
\TCons_\open, \TCons_{\mathsf{C}_4}, \TCons_{\mathsf{C}_2},
\TCons_{\mathsf{C}_3}, \TCons_{\mathsf{C}_3} \mr.$ Now, for checking
whether a bag belongs to the language of a regular bag expression, we
use either the algorithm from~\cite{jose-rdf-validation} based on
derivatives of regular expressions, or we use a slight modification of
the Interval algorithm presented on Figure~4 in
\cite{icdt2015}. Because the Interval algorithm supports only {\sf
  [0;1], [0; \unbound{}], [1; \unbound{}]} repetitions on
sub-expressions that are not triple constraints, we need to
\emph{unfold} repetitions that are not of this form. For instance,
$\Expr$[2;4] is to be replaced by this grouping expression whenever
$\Expr$ is not a triple constraint: $\Expr, \Expr, \Expr[0;1],
\Expr[0;1]$. After the unfolding, we can apply the Interval algorithm.

\paragraph{On the Complexity of Validation} In \cite{icdt2015}, we
showed that validation of ShEx schemas is NP-complete. Note that
validation remains in NP with the new constructs defined in the
present paper. The high complexity is due to verifying whether the
neighbourhood of a node locally satisfies a shape definition. In the
algorithm presented here, checking whether a candidate map is a local
witness is polynomial if the modified Interval algorithm is used, but
there is an exponential number of candidates to be considered (see
below). Note also that given $\typing,\lw$, it is trivial
(polynomial) to verify whether this is a global typing witness.

\subsection{Implementation and Optimization Guidelines}
A first, easily avoidable source of complexity is the computation of
$\typingcert, \lwcert$, to be fed as input of the flooding
algorithm. These can be computed using the algorithm sketched in the
proof of Theorem~\ref{thm:certain-typing}. An optimized implementation
should however compute them on the fly and on demand. This can be
performed using a version of the flooding algorithm, for which we give
here some guidelines. If a test $(\node, \mathsf{S}) \in \typingcert$
or $(\node, \Neg\mathsf{S}) \in \typingcert$ is required and either
$(\node, \mathsf{S})$ or $(\node, \Neg\mathsf{S})$ is in the already
computed portion of $\typingcert$, we can answer that test right away. If none of
the latter has been computed so far, we have to call
$\mathsf{FloodingValidation}$ with $\typing_0 = \{(\node,
\mathsf{S})\}$. After the call returns, either $(\node, \mathsf{S})$
is in the result typing, then we add it to $\typingcert$, or
$\ULW(\node,\mathsf{S})$ is empty, and then we add $(\node,
\Neg\mathsf{S})$ to $\typingcert$. The function is recursively called
if another test involving $\typingcert$ is required during its
computation.

Another source of complexity is the computation of $\ULW$; this is
also the unique reason for non-tractability of validation. The size of
a the set $\CandULW(\node, \mathsf{S})$ can be exponential in the
number of repetitions of a property in a shape definition (where extra
properties are considered as repetitions). For instance, on the
Running Example, the set $\CandULW(\node_{\textsf{ex:issue1}},
\text{{\sf <IssueShape>}})$ contains 27 ($= 3^3$) candidate mappings
elements. All these have to be checked as potential elements of
$\ULW(\node, \mathsf{S})$. 
%
Therefore, decreasing the size of the $\CandULW$ sets is a critical
optimization, and can be obtained by decreasing the size of the
$\potTC(\node)$ sets. For that, we propose to use look-ahead
techniques. Continuing on Example~\ref{ex:local-witness}, the idea is
to remove $\TCons_{\mathsf{C}_3}$ from $\potTC(\edge_4)$ because {\sf
  ex:emin}, the target node of $\edge_4$, does not have a {\sf
  is:experience}, thus cannot satisfy {\sf
  <ProgrammerShape>} required by $\TCons_{\mathsf{C}_3}$. A \emph{look-ahead} for an edge
$\edge = (\node, \mathsf{q}, \node')$ and a triple consumer
$\TCons_{\mathsf{C}_i}$ consists in inspecting \emph{only} the
neighbourhood of $\node'$ trying to prove that $\node'$ \emph{does
  not} satisfy some shape $\mathsf{S}$ that is a conjunct in
$\mathsf{C}_i$, thus allowing to not add $\TCons_{\mathsf{C}_i}$ to
$\potTC(\edge)$. 
Look-ahead can be extended to two-look-ahead, three-look-ahead, etc.,
consulting a bit farther in the neighbourhood. As a future work, we
plan to develop static analysis methods for schemas that allow to
define useful look-ahead criterions. Such methods would identify
e.g. the properties that are always required by a shape, to be
looked-up first during look-ahead.

Finally, the $\ULW$ sets do not need to be stored and can be accessed
through an iterator. Every $\CandULW(\node, \mathsf{S})$
is defined as a Cartesian product, so it is easy to iterate on it. On
line \ref{lalgo:flooding-get-first} of
Algorithm~\ref{algo:backtracking}, it is enough to retrieve elements from
$\CandULW(\node, \mathsf{S})$, through its iterator,
until a local witness in $\ULW(\node, \mathsf{S})$ is found.

\subsection{Post-Validation Processing}
The global typing witness computed by the flooding algorithm
associates with every node the shape labels that it satisfies (in
$\typing$), and with every edge in the neighbourhood of a node, how it
participated in satisfying a shape (in $\lw$). This allows for
post-processing of the graph depending on the "roles" played by the
different nodes and edges. For instance, on the Running Example, we
could check the additional constraint "all user that reported an issue
is affected by that same issue". Another interesting use case is exporting
in e.g. XML format all confirmed issues together with the testers that
reproduced them. Exporting in XML and in JSON can be currently
performed by two existing modules\footnote{GenX
  \url{http://w3.org/brief/NDc1}}\footnote{GenJ
  \url{http://w3.org/brief/NDc2}}, implemented using semantic actions
fired after the validation terminates. Additionally, the global typing
witness can be exported using a normalized JSON
format\footnote{\url{http://shexspec.github.io/primer/ShExJ\#validation}},
thus making post-processing possible using virtually any programming
language.


\section{On error identification}
\label{sec:errors}
One of the uses of error identification
is to guide the user in rendering the input graph
into one that is valid i.e., repairing the graph.
\begin{example}
  \label{ex:repairing}
  Take the following RDF graph and the schema from the Running Example. 
  \begin{code}
    \begin{tabbing}
      \hspace{0.6cm} \= \hspace{0.4cm} \=\\
      ex:issue is:reportedBy ex:emma ; \\
      \> is:reproducedBy ex:ron, ex:leila . \\
      ex:emma foaf:name "Emma" ; is:experience is:senior . \\
      ex:ron foaf:name "Ron" ; is:role is:someRole . \\
      ex:leila foaf:name "Leila" ; is:experience is:junior ; \\
      \> is:clientNumber 3 ; is:affectedBy ex:issue .
    \end{tabbing}
  \end{code}
  The node {\sf ex:issue} does not satisfy {\sf <IssueShape>} because
  it does not have a property {\sf is:reportedBy} whose object is a client. 
  One can identify a number of possible scenarios explaining the
  invalidity of the RDF graph. One is that \textsf{ex:emma} is missing
  a {\sf is:clientNumber} property and such should be added. Another
  is that the triple {\sf ex:issue is:re\-por\-tedBy ex:emma} uses the
  wrong property and should be replaced by \textsf{ex:issue
    is:reproducedBy ex:emma}. Naturally, this would make
  \textsf{ex:issue} not having the required \textsf{is:re\-por\-tedBy}
  property. To satisfy this requirement two actions are possible:
  replacing \textsf{ex:issue is:reproducedBy ex:leila} by
  \textsf{ex:issue is:reportedBy ex:leila} since \textsf{ex:leila}
  satisfies \textsf{<UserShape>} and {\sf <ClientShape>}, or adding a
  new node satisfying \textsf{<UserShape>} and \textsf{<ClientShape>},
  and connecting \textsf{ex:issue} to the new node with
  \textsf{is:reportedBy}.\qed
\end{example}


Given a schema $\mathsf{Sch}$, a graph $\mathsf{G}$, and an initial
typing $\typing_0$ of $\mathsf{G}$ w.r.t.\ $\mathsf{Sch}$, we attempt
to find a graph $\mathsf{G}'$ obtained from $\mathsf{G}$ with a
minimal set of triple insertions and deletions such that $\mathsf{G}'$
satisfies $\mathsf{S}$ w.r.t.\ $\typing_0$ (replacing the property of
an edge consists of deleting and inserting an edge). Basically, we
wish to present the user a minimal set of operations that render the
input graph valid thus pinpointing possible reasons why $\mathsf{G}$
is not valid. Such a graph $\mathsf{G'}$ is called a \emph{repair} of
$\mathsf{G}$ w.r.t. $\mathsf{Sch}$. Unfortunately, the number of
different ways of repairing a graph may be (exponentially) large which
renders constructing a repair intractable.

\begin{example}
  \label{ex:repair-example}
  Take the following instance of RDF  
  \begin{code}
    \begin{tabbing}
      \hspace{0.6cm} \= \hspace{0.4cm} \=\\
      ex:term ex:has-vars ex:vars .\\
      ex:vars ex:x1-t "x1-true" ; ex:x1-f "x1-false" ; \\
      \> ex:x2-t "x2-true" ; ex:x2-f "x2-false" . 
    \end{tabbing}
  \end{code}
  and consider the setting where we wish to verify if the node
  \textsf{ex:term} satisfies the type \textsf{<Term>} of the following
  schema 
  \begin{code}
    \begin{tabbing}
      \hspace{0.6cm} \= \hspace{0.4cm} \=\\
      <Term> \{ ex:has-vars ex:vars \}\\
      <Vars> \{\\
      \> (ex:x1-t xsd:string | ex:x1-f xsd:string | \EmptyShape), \\
      \> (ex:x2-t xsd:string | ex:x2-f xsd:string | \EmptyShape) \}
    \end{tabbing}
  \end{code}
  The repairs of the above RDF instance correspond to the
  set of all valuations of two Boolean variables \textsf{x1} and
  \textsf{x2}, which is extensible to an arbitrary
  number of variables. Although the schema also permits an \emph{empty
    valuation} instance where all the outgoing edges of
  \textsf{ex:vars} are removed, such an instance is not a repair
  because it is not minimal. \qed
\end{example}
With additional shape definitions in the
schema and additional nodes in the RDF graph, one can encode
satisfiability of CNF formulas. 
According to Theorem~\ref{thm:1} (given here without proof), 
constructing a repair is unlikely to be polynomial (unless
P=coNP).
\begin{theorem}
  \label{thm:1}
  Checking if a given graph $\mathsf{G}'$ is a repair of a given graph
  $\mathsf{G}$ w.r.t.\ a given schema $\mathsf{S}$ and a given initial
  typing $\typing_0$ is coNP-complete.
\end{theorem}
%
Remark that the hardness proof uses only simple schema operators: the
some-of operator for encoding disjunction, and the grouping operator
for encoding conjunction. Thus, any schema language proposing these or
similar operators would have non-tractable repair problem. 

Error identification could be approached in a much simpler fashion:
rather than repairing the problem just point to the node(s)
responsible for the problem is and let the user deal with it. This
approach is, however, inherently ambiguous as already shown in
Example~\ref{ex:repairing}, and 
shows the necessity in developing suitable heuristics for error
identification.

\section{Related work}
\label{sec:related}

\paragraph{Recursive Validation Language and Cyclic Validation}
\label{sec:recurse}
In ShEx, a $\ShapeDefinition$ has a $\ShapeExpr$ composed of $\TripleConstraint$s with $\ValueClass$es.
In the case that a $\ValueClass$ contains a $\ShapeLabel$, the grammar becomes recursive because that $\ShapeLabel$ references a $\ShapeDefinition$.
For example, an {\sf <IssueShape>} could have an {\sf is:related} property which references another {\sf <IssueShape>}.
If the instance graph has cycles on edges which appear in shapes, validation may arrive at validating the same shape against the same node.
The schema languages Description Set Profiles and Resource Shapes described below haven't considered the problem of detecting or terminating cyclic validation.

\paragraph{Global Constraints}
\label{sec:glob}
Where ShEx focuses on typings of specific instance nodes by shapes, some schema languages define validation for RDF graphs as a whole.
This involves some variant of an iteration across nodes in the graph and shapes in the schema to perform a maximal typing.
Description Set profiles includes global cardinality constraints describing the number of permissible instances of specified shapes.
 ShEx's global maximal typing\cite{icdt2015} types all nodes with all the shapes they satisfy.

\paragraph{Description Set Profiles}
\label{sec:dsp}
The Dublin Core® Metadata Initiative is developing a constraint language called Description Set Profiles (DSP) \cite{dublin-core}. 
DSP can be representated by an RDF vocabulary and by a conventional XML Schema.
DSP  has additional value constraints for encoding, language tag lists, and specific rules for subproperties.

DSP does not address repeated properties though a 2008 evaluation  of DSP for a ``Scholarly Works Application Profile'' specifically identified a need for repeated properties e.g. dc:type\footnote{%
\url{http://tinyurl.com/eprint-dc-type1}}
\footnote{\url{http://tinyurl.com/eprint-dc-type2}}.
DSP's current interpretation of repeated properties treats them as conjunctions of constraints.
The study found that professional modelers had expected a behavior more like ShEx, i.e. each of the constraints would have to be individually matched by some triples in the neighborhood.

\paragraph{Resource Shapes}
\label{sec:rs}
ShEx was originally created to provide a domain-specific language for Resource Shapes \cite{resource-shapes}.
Resource Shapes is an RDF vocabulary for describing simple conjunctions of shape constraints.
While the specification was not clear on this, the author verbally indicated that repeated properties were probably not permitted.
Resource Shapes includes descriptive features {\sf oslc:readOnly}, {\sf oslc:hidden} and {\sf oslc:name} which have no effect on validation but can be useful to tools generating user input forms.

\paragraph{OWL Based Validation}
Another approach proposed for RDF validation was to use OWL to express constraints.  
However, the use of Open World and Non-unique name assumption limits validation possibilities. \cite{ClarkSirin13,Tiao10,Motik07} propose the use of OWL expressions with a Closed World Assumption to express integrity constraints. 
The main criticism against such an approach is that it associates an alternative semantics with the existing OWL syntax, which can be misleading for users.
Note that in any case, OWL inference engines cannot be used for checking the constraints, and such an approach requires a dedicated implementation.

\paragraph{SPARQL Based Validation}  
It is possible to use SPARQL to express validation constraints although the 
SPARQL queries can be long and difficult to debug so there is a need for a higher-level language.
SPARQL Inferencing Notation (SPIN)\cite{SPIN11} constraints associate {\sc RDF} types or nodes with validation rules. These rules are expressed as SPARQL queries. 
There have been other proposals using SPARQL combined with other technologies,
Simister and Brickley\cite{Simister13} propose a combination
 between SPARQL queries and property paths employed at Google. 
 Kontokostas et al~\cite{kontokostasDatabugger} proposed \emph{RDFUnit} a Test-driven framework which employs SPARQL query templates and Fischer et al~\cite{Fischer2015} propose RDF Data Descriptions, a domain-specific language that is compiled into SPARQL queries.
SPARQL has much more expressiveness than Shape Expressions and can even be used to validate numerical and statistical computations~\cite{Labra13}. On the other hand, SPARQL does not allow to support recursive constraints and the additive semantics of grouping is difficult to express in SPARQL. Reuter et al~\cite{Reuter2015} have recently proposed an extension operator to SPARQL to include recursion. With such operator, it might be possible to compile ShEx to SPARQL.

\paragraph{SHACL}
\label{sec:shacl}
The SHACL language is under development and the Working Group has several open issues related to it's differences with ShEx, most notably Issue 92\footnote{{\url{https://www.w3.org/2014/data-shapes/track/issues/92}}} related to the interpretation of repeated properties.
While still far from representing consensus in the group, the First Public Working Draft \cite{shacl} includes a core RDF vocabulary similar to but more expressive than Resoruce Shapes for describing shapes constraints.
The other part of the specification includes a SPARQL template convention and an algorithm for iterating through a graph and its constraints.
This template system implements the core semantics and, in principle, provides an extensibility mechanism to extend the vocabulary to features which can be verified in atomic SPARQL queries with a supplied subject.
Regarding recursion, it is hoped that SHACL will adopt some definition of well-defined schema which will enable sound recursion.
To the extend they are defined, SHACL's AND and OR constructs are analogous to ShEx's some-of and group (so long as there are no triples in the instance data which could match more than one repeated property).
The property constraints attached directly to a SHACL shape appear to have the same behavior as if they were inside an AND construct.
There was no schema for SHACL (a so-called SHACL for SHACL) published with the first published working draft.
In principle, a sufficiently constraining schema would accept only inputs for which there was a defined semantics.
This could provide an anchor for semantic definitions analogous to the role typically performed by an abstract syntax.

\section{Conclusions and future work}
\label{sec:concl}
ShEx is an expressive schema language for RDF graphs. We illustrated
the features of the language with examples, described its semantics,
and presented a validation algorithm.
ShEx has several open source implementations, and several
documentation resources available on the Web. It has been used for the
description and the validation of two linked data portals
\cite{Labra-ldqsemantics}. ShEx is currently successfully used in
medical informatics for describing clinical models.
ShEx represents a substantial improvement over contemporary sche\-ma languages in features and sound semantics.

As future development of ShEx, we are working on the definition of
high-level logical constraints on top of ShEx schemas, on data
exchange and data transformation solutions based on ShEx, and on
heuristics for helpful error reporting.

\end{document}